\documentclass[letterpaper, 10 pt, conference]{ieeeconf}  

\IEEEoverridecommandlockouts                              
\usepackage{graphics} 
\usepackage{epsfig} 
\usepackage{amsmath} 
\usepackage{algorithm}
\usepackage{amssymb}
\usepackage{balance}
\usepackage{romannum}
\usepackage[utf8]{inputenc}
\usepackage[english]{babel}
\usepackage[T1]{fontenc}

\newtheorem{theorem}{Theorem}
\newtheorem{corollary}[theorem]{Corollary}

\newtheorem{proposition}[theorem]{Proposition} 
\newtheorem{definition}{Definition}
\usepackage[noend]{algpseudocode}
\title{\LARGE \bf
Cooperative Energy Scheduling for Microgrids \\ under Peak Demand Energy Plans}

\author{Amir Valibeygi and Raymond A. de Callafon
\thanks{Amir Valibeygi is with the Department of Mechanical and Aerospace Engineering,
        University of California, San Diego
        {\tt\small avalibey@ucsd.edu}}%
\thanks{Raymond A. de Callafon is a professor with the Department of Mechanical and Aerospace Engineering,
        University of California, San Diego, La Jolla, CA 92094-0411, USA
        {\tt\small callafon@ucsd.edu}}%
}

\begin{document}

\maketitle
\thispagestyle{empty}
\pagestyle{empty}


\begin{abstract}
A cooperative energy scheduling method is proposed that allows joint energy optimization for a group of microgrids to achieve cost savings that the microgrids could not achieve individually. The discussed microgrids may be commercial entities in a distribution network under utility electricity rate plans comprising both Time of Use (ToU) and peak demand charge. Defining a stable operation as a situation where all microgrids would be willing to participate, it is shown that under such rate plans and in particular due to the peak demand charge, a cost distribution that is seemingly fair does not necessarily result in a stable cooperation. These results are derived in this paper using  concepts from cooperative games. It is therefore sought to devise a stable cost distribution algorithm that, while maximizing some measure of fairness among the participating microgrids, ensures they all benefit from their participation. A simple case study is presented that demonstrates fairness and stability aspects of the cooperation.
\end{abstract}


\section{Introduction}

For many small and medium-size commercial and industrial entities, electric utility charges involve two main components, one accounting for the entity's weighted hourly consumption (Time of Use or ToU) and the other for its peak demand during the entire billing interval (demand charge). ToU pricing imposes different unit prices at different times of the day to urge the users to shift their loads from certain peak hours to off-peak hours. On the other hand, peak demand charge urges the users to flatten their overall demand profile to realize smaller peak to average ratio.

Currently in California, peak demand charges account for up to 50\% of some industrial users’ monthly electricity costs \cite{neubauer2015deployment}. Users’ response to such pricing is dependent on their flexibility for time-shifting their loads as well as their energy storage capacity. While energy storage can be exploited by users to optimally schedule their consumption and reduce their costs, it can also benefit the grid by increasing its reliability under periods of high demand \cite{denholm2018potential}. Optimal scheduling of energy storage is a major topic of interest for electricity users of different sizes. For some major contributions see \cite{chen2012optimal,chandy2010simple}. While each user with energy storage can optimize its storage schedule, joint energy scheduling and optimization between multiple microgrids may further reduce their total cost. This can be roughly attributed to the fact that by joint scheduling, users who under-utilize their energy storage at some or all parts of the day can make it available to other users who need it at those times. The role of cooperative optimization in reducing intermittency and uncertainty has also fueled the interest in cooperative optimization methods \cite{rahbar2016energy}. Different works have studied cooperative management of loads and storage and joint optimization between multiple users
\cite{mohsenian2010autonomous,ananduta2018resilient}. In \cite{rahbar2016energy}, the authors consider the problem of cooperative energy management for two microgrids with ToU energy cost aiming to achieve total cost reduction as a result of cooperation. Shared ES management for a group of microgrids with profit coefficient sets is considered in \cite{rahbar2016shared}. The authors of \cite{nguyen2013optimal} propose a stochastic formulation of the cooperative energy scheduling problem for a group of microgrids and consider ToU as well as operating costs of their local energy resources. The work of \cite{ouammi2015coordinated} proposes a model predictive framework for cooperative optimization of a network of interconnected microgrids and discusses the attained total cost saving for all users.

While the approaches proposed in the earlier works present feasible strategies for total cost saving by a group of users under cooperation, they do not discuss distribution of the obtained total savings between the participating users. It is desired to investigate mechanisms that make such exchange profitable for all participating parties so as to encourage participation. The work of \cite{dehghanpour2017agent} discusses cooperative power management between multiple microgrids with renewable generation and proposes a pareto-optimal solution using Nash bargaining that encourages microgrids to participate. In \cite{baharlouei2013achieving}, the authors discuss fairness in optimal coordination of multiple users with quadratic energy cost. Also, the work \cite{nguyen2017bi} studies cooperative management in the wholesale electricity market and presents cost allocation solutions that have certain favorable properties. In \cite{lee2014direct}, the authors have used a cooperative game approach to tackle direct energy trading between DERs and energy consumers. In general, division of the attained cost saving between participating users in a fair and stable manner constitutes the major issue in the considered cooperative games \cite{saad2012game}.

In this work, we are interested in a specific yet prevailing case of cooperative energy scheduling; the case of multiple microgrids under ToU and demand charge energy plans. Cost distribution when the cost structure includes peak demand terms presents important stability implications that have not yet been addressed to the best knowledge of the authors. We will show that first, under such cost structure, joint scheduling will result in reduced total cost. Next, we will investigate how some cost distributions may de-stabilize the collaboration due to the existence of peak demand term in the cost function. An alternative algorithm is then proposed that can provably guarantee the benefit of all users from participation in the collaboration while maximizing some measure of fairness among them. Our main contribution lies in the consideration of demand charge in the cost structure which leads to non-submodularity of the cost function and therefore necessitates careful consideration of stability. We will show that for this problem, a stable distribution of the optimal cost between the users that is desirable from all users’ standpoints will always exists and will provide an approach for computing such distribution.

\section{Preliminaries and Benchmark Problem}
Consider a distribution system with a single energy provider (utility) and a set of microgrids $\mathcal{N}=\{1,2,...,N\}$ where each microgrid $n\in \mathcal{N}$ is equipped with a smart meter and means of two-way data communications with the utility and other users. The terms \textit{user} and \textit{microgrid} are used interchangeably throughout the paper. Some users may additionally possess a battery energy storage (ES) system. The set $C=\{C_1,C_2,...,C_N\}$ indicates the energy storage capacity of each user. We further take the set of time intervals $\mathcal{T}=\{1,2,...,T\}$ where each $t\in \mathcal{T}$ has length $\Delta_T$ to represent the energy scheduling horizon. The energy demand vector for each user $n\in \mathcal{N}$ is defined as $d_n\triangleq \begin{bmatrix}
d_n^1 , d_n^2 , ... , d_n^T
\end{bmatrix}$ while energy flow from the grid to user $n$ is $x_n\triangleq \begin{bmatrix}
x_n^1 , x_n^2 , ... , x_n^T
\end{bmatrix}$. The sum of energy supplied to each user from the main grid and from the user's ES during interval $t \in \mathcal{T}$ should equal $d_n^t$. The energy provided by the ES is denoted by $e_n^t$. 
Each user $n\in \mathcal{N}$ with storage capability can supply part or all of its demand according to
\begin{align}
d_n^t=x_n^t+e_n^t
\label{eq:energy_balance}
\end{align}
and its storage charge level $c_n$ varies as
\begin{align}
c_n^{t+1}=c_n^t-e_n^t
\label{eq:storage_charge}
\end{align}
Such user can make dispatch decisions for storage charge/discharge ($e_n^t$) to optimize its electricity consumption subject to the following storage constraints
\begin{align}\nonumber
&e_n^{min}\leq e_n^t \leq e_n^{max}\\
&c_n^{min}\leq c_n^t \leq c_n^{max}\\ \nonumber
&c_n^{t_0} = c_n^{t_{end}}
\end{align}
While this simplified model is adopted to only capture salient dynamics of the system and demonstrate the pivotal ideas of this work, the methods are mostly applicable under more complex microgrid models.\\
The energy consumers in this study are commercial and industrial users that are billed under both ToU and demand charge pricing plans.
The total cost of energy for user $n$ under such rate plans can be computed as
\begin{align}
f_n(x_n)=\sum_{t=1}^{T}p^t x_n^t + \alpha~\underset{t \in \mathcal{T}}{\max}~x_n^t
\label{eq:cost_p1}
\end{align}
where $p^t$ is the ToU unit price at time $t$ and $\alpha$ is the demand charge coefficient.


\textit{\textbf{Optimization 1.}} The cost minimization problem for each individual user having ES can be formulated as
\begin{align}
&\underset{x_n}{\mathrm{minimize}}~f_n(x_n) \label{eq:problem_1} \\
&\nonumber \mathrm{subject~to~} (1-3)
\end{align}
The optimal solution $x_n^*$ to this problem will result in the optimal cost $f_n(x_n^*)$. The above problem has a convex objective function and affine constraints and therefore is a convex program \cite{boyd2004convex}.  If each user solves its energy optimization (\textit{optimization 1}) individually without cooperation with other users, the total cost of all users would become
\begin{align}\nonumber
f_{non-coop}= \sum_{n=1}^{N}{f_n(x_n^*)}=\sum_{n=1}^{N}{\Big[\sum_{t=1}^{T}p^t {x_n^t}^* + \alpha~\underset{t \in \mathcal{T}}{\max}~{x_n^t}^*\Big]
}
\end{align}
Next, we consider how a few users may cooperate to reduce this overall cost.


\section{Cooperative Optimization}
\subsection{Motivation}

Medium and large-scale energy consumers have diverse demand patterns and often significantly high peak-to-average demand ratios. While demand charge is a major part of total energy cost for such consumers, substantial demand charges may be incurred due to uneven demand distribution and high peak-to-average ratio \cite{neubauer2015deployment}. The addition of storage for reducing peak demand could significantly reduce demand charge in such cases. 
Although individual microgrids can utilize their ES to optimize their demand as described in the previous section, joint storage utilization between multiple users in a distribution network can potentially bring additional benefits. This can be attributed to the benefits gained by shared storage utilization as well as the fact that sum of users peaks is always greater than or equal to the peak of the collective consumption of all users. We propose a method of joint optimization between multiple microgrids with identical cost structures, such that the group of collaborating microgrids, although physically remote, will purchase electricity from the main grid as a whole. This can be further conceptualized by considering an aggregator that manages interaction and cooperation of participants and represents them as a whole.


The main questions that we are after answering in this section are the following:

1. Does forming a cooperation between the microgrids under the described cost structure result in attaining a lower total cost?

2. How should the savings attained as a result of cooperation be distributed between the participants?

We devote the rest of this section to answering these questions.

\subsection{Centralized Cooperative Optimization}

Consider a group of microgrids some of which having energy storage and suppose that energy flow between the main grid and the group is measured at a virtual point of connection. Denoting the power flow at this point by $x$, the total cost incurred by such group of users is
\begin{align}
f_{coop}(x)=\sum_{t=1}^{T}p^t x^t + \alpha~\underset{t \in \mathcal{T}}{\max}~x^t
\label{eq:cost_p2}
\end{align}

\textit{\textbf{Optimization 2.}} The cooperative optimization problem can be formulated as
\begin{align}
&\underset{x}{\mathrm{minimize}}~f_{coop}(x)
\label{eq:problem_2}
\end{align}
subject to the constraints
\begin{align}\nonumber
&x^t=\sum_{n=1}^{N}{d_n^t}-\sum_{n=1}^{N}{e_n^t}\\
&c_n^{t+1}=c_n^t-e_n^t~~\forall n \in \mathcal{N} \label{eq:constraints_p2}
 \\ \nonumber
&e_n^{min}\leq e_n^t \leq e_n^{max}~~\forall n \in \mathcal{N} \\ \nonumber
&c_n^{min}\leq c_n^t \leq c_n^{max}~~\forall n \in \mathcal{N}
\end{align}
The optimal cost resulting from this optimization is denoted by $f_{coop}^*$. Similar to \textit{optimization 1}, this problem is a convex program and can be solved using standard techniques \cite{boyd2004convex}.

\subsection{Cooperative Game Framework for User Collaboration}
We use concepts from cooperative games to study how different users may form coalitions to attain payoffs that would not be possible if they were to optimize their demand individually and how this increased payoff should be divided between them such that it satisfies some measures of fairness and stability. In this section, we introduce a few concepts from cooperative games including the notions of fairness and stability and will apply them to the cooperative energy optimization problem.

Consider the same set of users $\mathcal{N}$ as the previous section. Any nonempty subset $S\subseteq \mathcal{N}$ is called a coalition between the members of $S$. A coalitional cost game between the members of $\mathcal{N}$ is a pair $(\mathcal{N},v)$ where $\mathcal{N}$ is the set of users and $v(S):2^\mathcal{N}\rightarrow \mathbb{R}^+$ is a set function representing the cost of each coalition $S\subseteq \mathcal{N}$. In a cost game, users prefer less cost and therefore may form coalitions to reduce the total cost. 
\begin{definition}
A cooperative game $(\mathcal{N},v)$ is called an LP game if its set cost function can be expressed as $v(S)=\underset{x}{\mathrm{min}}~u^Tx$ subject to linear constraints on $x$.
\end{definition}

We define the cost of each coalition $S$ as $v(S)=f^*_{coop,S}$. Here $f^*_{coop,S}$ is the optimal cost resulting from the coalition of users of set $S$. The coalition $\mathcal{N}$, i.e. the coalition of all users in $\mathcal{N}$ is called the grand coalition. We also define $\psi \in \mathbb{R}^N$ as the vector of cost distribution between all users in the grand coalition with each user's cost being $\psi_i$.

\begin{definition}
A game $(\mathcal{N},v)$ is said to be \textit{sub-additive} if given that $S \cap T = \phi$ it results that $v(S \cup T) \leq v(S) + v(T)$.
\end{definition}

\begin{proposition}
The above game $(\mathcal{N},v)$ with $v(S)=f_{coop,S}^{*}$ is sub-additive.
\end{proposition}

\begin{proof}
Suppose $S$ and $T$ are two disjoint sets of users and $f_{coop,S}^*$ and $f_{coop,T}^*$ are their respective optimal costs. Now assume these two sets of users join to form the coalition $S \cup T$. If in the joint coalition, both sets of users $S$ and $T$ maintain their optimal schedules $x_S^*$ and $x_T^*$ from the optimal solution to \textit{optimization 1} (this is possible because $S$ and $T$ are disjoint), then
\begin{align}\nonumber
&f_{coop,S \cup T}=\sum_{t=1}^{T}p^t (x_S^t+x_T^t) + \alpha~\underset{t \in \mathcal{T}}{\max}~(x_S^t+x_T^t)\\ \nonumber
& \leq \sum_{t=1}^{T}p^t x_S^t + \alpha~\underset{t \in \mathcal{T}}{\max}~x_S^t+ \sum_{t=1}^{T}p^t x_T^t + \alpha~\underset{t \in \mathcal{T}}{\max}~x_T^t\\ \nonumber
&=f_{coop,S}^*+f_{coop,T}^*
\end{align}
Therefore, for the above joint schedule, we will always have $f_{coop,S \cup T} \leq f_{coop,S}^*+f_{coop,T}^*$. Since we always have $f_{coop,S \cup T}^* \leq f_{coop,S \cup T}$, it results that $f_{coop,S \cup T}^* \leq f_{coop,S}^*+f_{coop,T}^*$ or $v(S \cup T) \leq v(S)+v(T)$.
\end{proof}

This property implies that the total saving increases as more users participate in the game. It can also be intuitively verified that since the peak demand of the unified entity over a time horizon will always be equal or smaller than the sum of peak demands of all users, if all users maintain their optimal individual schedule without considering other users' schedules, the cost of the unified entity will be at worst equal to the sum of individual costs of all users. We have therefore found the answer to question \Romannum{1}. While the sub-additivity property discusses the total saving of all participating users, it does not discuss the distribution of savings between the users. The next section will address the distribution of the total optimal cost between the users.

\subsection{Distribution of saving from Cooperative Optimization}
 So far we have shown that the cost achieved as a result of cooperation between a group of users is at worst equal to the sum of costs of individual users if no cooperation took place. However, this does not necessarily lead to all users paying a lower cost than they did individually. In fact, whether or not users will be subject to a lower cost depends on the distribution of the total cost between the users. Two important concepts often considered while specifying cost distribution are stability and fairness. We will further characterize these two notions and study our cooperative game with regards to fairness and stability.
 
\textbf{Fairness.} A fair cost distribution $\psi \in \mathbb{R}^N$ divides the total cost of a cooperation based on different users' contributions in achieving that cost. The well-known characterization of fairness given be Shapley \cite{shapley1953value} that attributes certain desirable properties to a fair distribution is given as
\begin{theorem}
\cite{shapley1953value} Given the coalitional game $(\mathcal{N},v)$ with $v(S)=f_{coop,S}^{*}$, the unique cost distribution $\psi(\mathcal{N},v)$ that divides the entire saving of the grand coalition between the users and satisfies fairness axioms of \cite{shapley1953value} is given by
{\footnotesize
\begin{align}\nonumber
\psi_n&(\mathcal{N},v)=\nonumber
&\frac{1}{|\mathcal{N}|!}\sum_{S\subseteq \mathcal{N} \setminus {n}}{|S| !~(|\mathcal{N}|-|S|-1)!\big[v(S\cup {n})-v(S)\big]}
\end{align}}
\end{theorem}
where $|\mathcal{N}|$ and $|S|$ are the cardinality of $\mathcal{N}$ and $S$ respectively. The fairness axioms and proof can be found in \cite{shapley1953value}. The outcome of this theorem assigns portion $\psi_n$ of the total saving to each user $n$, also known as the Shapley value of user $n$. The above computation involves determining the average marginal contribution of user $n$ over all the different ways that the grand coalition can be formed from the zero coalition. Therefore one should solve \textit{optimization 2} for all possible subsets of the grand coalition and then compute the marginal contribution of each user using the above relation.

\textbf{Stability.} Although the Shapley distribution guarantees fair and efficient distribution of the saving between the users in accordance with the given axioms, there is still no guarantee that such saving is binding for all users. In fact, we now aim to address the question of stability of the grand coalition: under the Shapley saving distribution given above and despite the fact that the game is sub-additive, could there be any incentive for a user or a group of users to refrain from joining the grand coalition and form smaller coalitions among themselves? To answer this question, we next introduce the concept of \textit{core} in coalitional games.

\begin{definition}
A cost distribution vector $\psi \in \mathbb{R}^N$ is in the core of a coalitional game $(\mathcal{N},v)$ if and only if 
\begin{align}
\sum_{n\in \mathcal{N}}{\psi_n = v(\mathcal{N})} ~ \textrm{and} ~ \sum_{s\in S}{\psi_s \leq v(S)}, \forall S \subset \mathcal{N}
\label{eq:core}
\end{align}
\label{definition:core}
\end{definition}

In other words, in order for a coalition to be stable, the sum of savings of any subset of users should be greater or equal to the saving that those users would have obtained if that subset actually formed a sub-coalition. It is now reasonable to assume that rational users participating in a coalitional game would want to form the grand coalition only if the saving distribution vector is drawn from the core. We can henceforth relate the concept of core with the stability of a coalitional game. This concept is similar to the concept of Nash equilibrium in cooperative games. The difference here is that instead of investigating whether a single user can benefit by deviating from an action, we turn our attention to possible deviation of a group of users in order to form their own coalition. It can be shown \cite{osborne1994course} that the core of a cooperative game may be empty. If a game has an empty core, no saving distribution can possibly guarantee the stability of the grand coalition. In such case, a player or group of players may opt for smaller coalitions to increase their savings. The following theorem will examine non-emptiness of the core for the game considered in this work.

\begin{corollary}
The core of the game $(\mathcal{N},v)$ with $v(S)=f_{coop,S}^{*}$ is non-empty.
\label{corollary:nonempty_core}
\end{corollary} 

\begin{proof}
Since $v(S)=f_{coop,S}^{*}=\underset{x}{\mathrm{minimize}}~f_{coop,S}(x)$ can be expressed as a linear program with the linear cost functions of \textit{optimization 2}, this game is an LP game. It can be proved \cite{curiel2013cooperative} that an LP game has a nonempty core and therefore the core of the above game is non-empty.
\end{proof}

It can be shown that if the core of a coalitional game exists, it may not be unique \cite{osborne1994course}. For certain class of games known as submodular games, if the cost is distributed according to Shapley distribution, then the grand coalition will be stable.
\begin{definition}
A game $(\mathcal{N},v)$ is said to be \textit{submodular} (concave) if $v(S \cup T) + v(S \cap T) \leq v(S) + v(T),~\forall S,T \subseteq \mathcal{N}$ or equivalently $v(S\cup \{i\})-v(S) \geq v(T\cup \{i\})-v(T), \forall S\subseteq T \subseteq \mathcal{N} \backslash \{i\}, \forall i \in \mathcal{N}$.
\label{definition:submodular}
\end{definition}
In the following, we will investigate the submodularity of the considered collaborative cost minimization game.
\begin{proposition}
The game $(N,v)$ with $v(S)=f_{coop,S}^{*}$ is not submodular.
\end{proposition} 

\begin{figure}[thpb]
    \centering
    \includegraphics[width=\columnwidth]{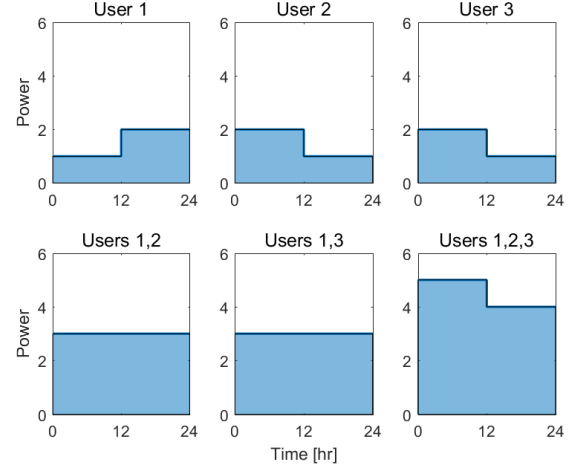}
    \caption{Users demand and consumption patterns without storage capacity. Top: Demand profiles for each of the users 1, 2, and 3. Bottom: Resulting demand profiles if coalitions form between different subsets of users.}
    \label{fig:nonconcavity}
\end{figure}

\begin{proof}
We will show that this game violates the submodularity property by a simplified example of a three-user cooperative game. Assume that three users none of which have storage capacity participate in a cooperative power sharing game. Also for simplicity assume that $p(t)=0,~ \forall t \in \mathcal{T}$, so only demand charge is applicable. Suppose the demand of each of the users follows the profiles of Fig. \ref{fig:nonconcavity}. Taking $S= \{1 \} $ and $T= \{1,2 \} $ as two possible coalitions between the users and $i=3$, we see that such coalitions satisfy $S \subseteq T \subseteq \mathcal{N}\backslash \{i\}$. Now, it can be seen from Fig. \ref{fig:nonconcavity} that
\begin{align}\nonumber
& v(S) = 2~,
v(S\cup \{i\})=3~,
v(T) = 3~,
v(T\cup \{i\})=5~ \\ \nonumber
& v(S\cup \{i\})-v(S)  = 1 \ngeq 2 = v(T\cup \{i\})-v(T)
\end{align}

Therefore the above game is not sub-modular according to definition \ref{definition:submodular}.
\end{proof}

It can further be shown that if the cost function contains onlly the ToU term, the game would satisfy the submodularity condition. The proof is eliminated due to space constraints. The lack of submodularity is therefore due to the existence of peak demand cost term. The lack of submodularity implies that the resulting Shapley value may fail to fall within the core of the game \cite{curiel2013cooperative}, as will be shown in the numerical case study of the next section. However, according to corollary \ref{corollary:nonempty_core}, the core of the game is non-empty in both cases and therefore, one should be able to find another cost distribution that falls within the core. In fact, according to definition \ref{definition:core}, all distributions $\psi_n$ that satisfy the inequalities (\ref{eq:core}) constitute the core.

\section{Fair and Stable Cost Distribution Algorithm}
Non-emptiness and non-uniqueness of the core imply that multiple payoff distributions might exist that are within the core. In the ideal case, the Shapley distribution is chosen over other possible distributions due to its fairness properties. However, if Shapley distribution is not within the core (is not stable), we would want to have another notion of fairness that helps us pick one distribution among all stable distributions. To this end, we would define our notion of fairness as the difference between the highest and lowest percentage cost saving among all users for a given distribution vector. Using this fairness index, the stable and fair payoff distribution problem is formulated as

\begin{align}\nonumber
\underset{\psi}{\mathrm{minimize}}~~ & \Delta \\ \nonumber
\textrm{subject to}~& \sum_{n\in \mathcal{N}}{\psi_n = v(N)} \\ 
& \sum_{s\in S}{\psi_s \leq v(S)},~\forall S \subset \mathcal{N} \label{eq:fair_dist_LP}
\\ \nonumber
&  \Lambda_{min} \leq \frac{v(\{n\})-\psi_n}{v(\{n\})}  \leq \Lambda_{max},~\forall n\in \mathcal{N}\\ \nonumber
& \Delta = \Lambda_{max} - \Lambda_{min}
\end{align}
In this LP, the first two constraints enforce stability while the next two constraints formulate a measure of fairness as the difference between maximum and minimum percentage cost saving among all users. The objective is then to minimize this difference subject to stability constraints. Since the game is proved to have a non-empty core, the solution of this LP will always exist and will fall within the core of the game. Based on this program, we formalize the following cooperative optimization and cost distribution algorithm between a group of users. It should be noted that since the resulting cost distribution will be drawn from the core, it is in the interest of all users regardless of their demand pattern and storage capacity to participate in this game.

\begin{algorithm}
\caption{Cooperative optimization and cost sharing}\label{algorithm}
\begin{algorithmic}[1]
\State Collect load data for all users
\State Solve the cooperative optimization (\textit{optimization 2}) with cost function (\ref{eq:cost_p2}) subject to the constraints (\ref{eq:constraints_p2})
\State Distribute the total cost between the users according to Shapley distribution
\State Check if the Shapley distribution falls within the core by checking conditions (\ref{eq:core})
\State \textbf{If} Shapley distribution is within the core \textbf{then}
\State \quad Distribute the cost according to Shapley distribution
\State \textbf{else}
\State \quad Compute a fair and stable distribution by solving the linear program (\ref{eq:fair_dist_LP})
\State \textbf{End if}
\end{algorithmic}
\end{algorithm}

By following this algorithm, a stable payoff distribution is computed so that all users will have an incentive to join and form the grand coalition. This cost distribution algorithm provides the answer to question \Romannum{2}.

\textbf{Remark.} The complexity of checking condition (\ref{eq:core}) grows exponentially with $|\mathcal{N}|$. Therefore for very large number of participating microgrids, the Shapley computation steps and its stability checks may be skipped. The proposed optimization (\ref{eq:fair_dist_LP}) can directly be used to obtain a stable distribution in such cases to ensure scalability of the cost distribution.

\section{Case Study}
Consider three users with load demand profiles shown in Fig.\ref{fig:case2_power} left and storage capacities of respectively $500,300,700~kWh$. The load profiles represent predictions for two industrial and one commercial facility. Suppose these consumers are under the same ToU and peak demand charge electric rate plan. Normalized ToU price tariffs are shown in Table \ref{tab:price}. Also, the normalized peak demand charge coefficient is equal to $10$. 

\begin{table}[thpb]
\caption{ToU unit prices ($p^t$)}
\label{tab:price}
\begin{center}
\begin{tabular}{ |c||c|c|c|c| } 
 \hline
 Time &9am-12pm&12pm-6pm&6pm-9pm&9pm-9am\\ \hline
 $p^t$  &1.5&2&1.5&1\\ \hline
\end{tabular}
\end{center}
\end{table}

\begin{figure}[thpb]
    \centering
    \includegraphics[width=\columnwidth]{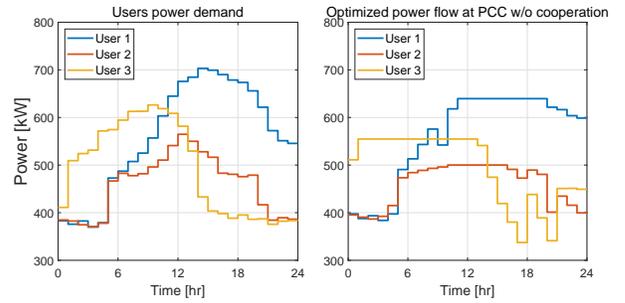}
    \caption{Left: Demand profiles for three industrial/commercial users. Right: Results of individual optimization (solution of  \textit{optimization 1})}    \label{fig:case2_power}
\end{figure}
Solving the convex optimization (\textit{optimization 1}) individually for each user results in power flow profiles of Fig. \ref{fig:case2_power}, right. Next, solving the cooperative scheduling problem (\textit{optimization 2}) for these users, the total cost of the system reduces from 67432 to 66174 as a result of joint scheduling. The optimal cost of other possible sub-coalitions are also shown in Table \ref{tab:case_2_cost}. Fig. \ref{fig:case2_total_power} compares total power flow of the three microgrids under individual versus joint scheduling. The flattened utility power profile under cooperative scheduling explains the lower cost achieved under the cooperation. Next, Shapley values are computed to obtain the share of each user in the total cost. It is observed that the following condition is violated
\begin{align*}
v(13)=45851 \ngeq 45873 = \psi_1 + \psi_3
\end{align*}
and therefore the Shapley distribution is not within the core of the game. Following algorithm 1, we then compute a fair cost distribution from the core by solving (\ref{eq:fair_dist_LP}). The resulting stable distribution is compared with other distributions in Fig. \ref{fig:case2_cost}. Satisfaction of the following inequalities verifies the presence of this distribution within the core and hence stability of the game.
\begin{align*}
&v(1) = 25522 \geq 24881 = \psi_1\\
&v(2) = 20399 \geq 20324 = \psi_2\\
&v(3) = 21510 \geq 20970 = \psi_3\\
&v(12) = 45806 \geq 45205 = \psi_1+\psi_2\\
&v(13) = 45851 \geq 45851 = \psi_1+\psi_3\\
&v(23) = 41587 \geq 41294 = \psi_2+\psi_3\\
&v(123) = 66174 = 66174 = \psi_1+\psi_2+\psi_3
\end{align*}

\begin{figure}[thpb]
    \centering
    \includegraphics[width=0.8\columnwidth,height=3.5cm]{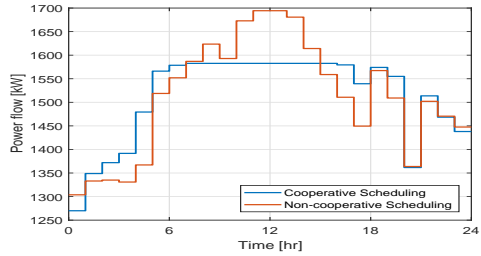}
    \caption{Result of cooperative optimization between all users (blue) vs. sum of all users' consumption under individual optimization (red).}
    \label{fig:case2_total_power}
\end{figure}

\begin{figure}[thpb]
    \centering
    \includegraphics[width=0.8\columnwidth]{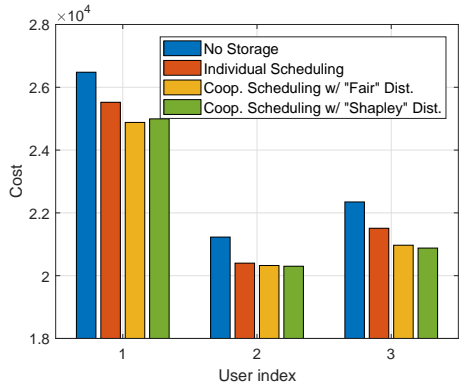}
    \caption{Cost of each user under individual optimization vs under different distributions of the coalitional optimization. Notice the reduction in cost if each user employs storage individually and also if users conduct cooperative scheduling. 
    }
    \label{fig:case2_cost}
\end{figure}

\section{Conclusion}
For a group of microgrids participating in energy plans involving peak demand charge, a cooperative energy scheduling algorithm is proposed that reduces the total energy cost as well as individual costs of each microgrid. Stability of the coalition is investigated by studying properties of the considered cost structure. To motivate participation among microgrids, a cost allocation algorithm is proposed that maximizes some measure of cost distribution fairness while satisfying stability properties of the distribution.

\begin{table}[thpb]
\caption{Optimal cost of \textit{optimization 2} under different possible sub-coalitions of the grand coalition}
\label{tab:case_2_cost}
\begin{center}
\begin{tabular}{ |c|c|c|c|c| } 
 \hline
 Coalition&\{1\}&\{2\}&\{3\}& \\ \hline 
 Cost& 25522  & 20399 & 21510 & \\ \hline \hline
Coalition&\{1,2\}&\{1,3\}&\{2,3\}&\{1,2,3\} \\ \hline
Cost& 45806 & 45851 & 41587 & 66174 \\ \hline
\end{tabular}
\end{center}
\end{table}


\bibliographystyle{IEEEtran}
\bibliography{IEEEabrv,mybib}

\end{document}